\documentclass[11pt]{article}
\usepackage{arxiv}
\usepackage{amsmath}
\usepackage{amssymb}
\usepackage{amsthm}
\usepackage{hyperref}
\usepackage{graphicx}
\usepackage[numbers]{natbib}

\theoremstyle{plain}
\newtheorem{theorem}{Theorem}
\newtheorem{proposition}{Proposition}

\theoremstyle{definition}
\newtheorem{definition}{Definition}
\newtheorem{assumption}{Assumption}

\theoremstyle{remark}
\newtheorem{remark}{Remark}

\title{Moravec's Paradox and Restrepo's Model: Limits of AGI Automation in Growth}

\author{
Marc Bara\thanks{\href{https://orcid.org/0009-0005-1480-5760}{ORCID: 0009-0005-1480-5760}} \\
Department of Strategy and General Management \\
ESADE Business School \\
Barcelona, Spain \\
\texttt{marcoantonio.bara@esade.edu}
}

\date{September 29, 2025}

\begin{document}
\maketitle

\begin{abstract}
This note extends Restrepo (2025)'s model of economic growth under AGI by incorporating Moravec's Paradox—the observation that tasks requiring sensorimotor skills remain computationally expensive relative to cognitive tasks. We partition the task space into cognitive and physical components with differential automation costs, allowing infinite costs for some physical bottlenecks. Our key result shows that when physical tasks constitute economic bottlenecks with sufficiently high (or infinite) computational requirements, the labor share of income converges to a positive constant in the finite-compute regime (rather than zero). This fundamentally alters the distributional implications of AGI while preserving the growth dynamics for cognitive-intensive economies.
\end{abstract}

\section{Introduction}
Restrepo (2025) develops a framework for economic growth in which Artificial General Intelligence (AGI) can perform any human task given sufficient computational resources. In his model, all economically essential ``bottleneck'' work is eventually automated, wages converge to the computational cost of replicating human work, and labor's share of GDP approaches zero as computational resources expand.

This note relaxes one of his assumptions: that all task types have uniform automation costs. Drawing on Moravec's Paradox \citep{moravec1988}—the observation that tasks humans find effortless (perception, mobility, manipulation) often require enormous computational resources, while tasks humans find difficult (mathematics, logic) require relatively modest computation—we extend his model to allow for differential automation costs across cognitive and physical tasks.

Empirical estimates support this differential: For cognitive tasks, modern language models like GPT-4 require approximately \(10^{11}\) to \(10^{13}\) FLOPS for inference on reasoning tasks \cite{patel2023gpt4, villalobos2023inference}. In contrast, real-time sensorimotor tasks in robotics, such as dexterous manipulation or navigation, demand \(10^{15}\) to \(10^{18}\) FLOPS per second due to high-fidelity physics simulations, visual processing, and control loops \cite{erdil2024moravec}. This ``Moravec gap'' of several orders of magnitude justifies our assumption of \(\alpha^p(\omega) \gg \alpha^c(\omega)\).

Our contribution is consequential: we show that when physical bottleneck tasks have sufficiently high computational requirements relative to available compute, they can remain human-operated over long horizons—and indefinitely if $\alpha^p=\infty$ or compute is bounded—preserving a positive labor share even as the economy grows. This modification reconciles AGI models with the persistent difficulty of robotic automation in physical domains.

Recent estimates place global AI compute capacity at approximately \(10^{22}\) FLOPS in 2025, dominated by the US and private sector \citep{pilz2025}.

\section{Model}

\subsection{Setup}

Following \citet{restrepo2025}, the economy produces output $Y_t$ by completing work tasks $\omega \in \Omega$. We extend his framework by partitioning the task space:
$$\Omega = \Omega_c \cup \Omega_p$$
where $\Omega_c$ represents cognitive tasks and $\Omega_p$ represents physical tasks, with $\Omega_c \cap \Omega_p = \emptyset$.

The quantity of task $\omega$ completed at time $t$ is:
$$X_t(\omega) = L_t(\omega) + \frac{1}{\alpha_t(\omega)}Q_t(\omega)$$
where $L_t(\omega)$ is human labor allocated to task $\omega$, $Q_t(\omega)$ is computational resources allocated to task $\omega$, and $\alpha_t(\omega)$ is the computational cost of replicating one unit of human work.

Output is produced according to:
$$Y_t = F(\{X_t(\omega)\}_{\omega \in \Omega})$$
where $F$ is increasing, differentiable, concave, and exhibits constant returns to scale.

\subsection{The Moravec Modification}

We depart from Restrepo by assuming differential automation costs:

\begin{assumption}[Differential Automation Costs]
\label{ass:moravec}
For cognitive tasks $\omega \in \Omega_c$, $\alpha_t^c(\omega) \to \alpha^c(\omega) < \infty$. For physical tasks $\omega \in \Omega_p$, $\alpha_t^p(\omega) \to \alpha^p(\omega)$, where $\alpha^p(\omega) \gg \alpha^c(\omega)$ and $\alpha^p(\omega) = \infty$ for a non-empty subset of physical bottlenecks to capture fundamental physical limits (e.g., energy or real-time constraints in robotics).
\end{assumption}

This captures Moravec's Paradox: cognitive tasks become automatable with finite (and relatively modest) compute, while physical tasks require vastly more computational resources or may remain impossible to automate.

\subsection{Bottlenecks and Resource Constraints}

Following Restrepo's Definition 1:

\begin{definition}[Bottleneck Work]
\label{def:bottleneck}
Task $\omega$ is a bottleneck if for any $\{X_t(\omega)\}$ such that output $F(\{X_t(\omega)\})$ is unbounded, either:
\begin{enumerate}
\item $X_t(\omega)$ is unbounded, or
\item The marginal product $F_\omega(\{X_t(\omega)\})$ is unbounded
\end{enumerate}
\end{definition}

We introduce:

\begin{assumption}[Physical Bottlenecks Exist]
\label{ass:physical}
There exists a non-empty subset $\Omega_p^B \subset \Omega_p$ of physical bottleneck tasks.
\end{assumption}

The economy faces resource constraints:
\begin{align}
\sum_{\omega \in \Omega} Q_t(\omega) &\leq Q_t \\
\sum_{\omega \in \Omega(s)} L_t(\omega) &\leq H(s) \quad \forall s \in \mathcal{S}
\end{align}
where $Q_t$ is total computational resources and $H(s)$ is the supply of skill type $s$.
\paragraph{Two regimes used in the results.}
We analyze two limiting regimes:
\begin{itemize}
\item \textbf{Finite-compute regime}: $Q_t \uparrow Q_{\max}<\infty$. Persistence of human labor can follow from large but \emph{finite} automation costs.
\item \textbf{Unbounded-compute regime}: $Q_t \to \infty$. Persistence follows if at least one essential physical bottleneck has $\alpha^p(\omega)=\infty$.
\end{itemize}
Each result below states which regime it requires.

\section{Results}

\subsection{Persistent Human Labor in Physical Tasks}

\begin{proposition}[Incomplete Automation: Finite-Compute Regime]\label{prop:incomplete}
Assume $Q_t \uparrow Q_{\max}<\infty$ and let $\omega\in\Omega_p^B$ be an essential physical bottleneck.
Let $L_\omega>0$ denote the asymptotic human labor flow allocated to $\omega$ along an efficient allocation.
If
\[
\alpha^p(\omega)\,L_\omega \;>\; Q_{\max},
\]
then $\omega$ cannot be fully automated in the limit: $L_t(\omega)>0$ for all sufficiently large $t$.
\end{proposition}

\begin{proof}
Suppose instead that $L_t(\omega)\to 0$ and $\omega$ is fully automated asymptotically. Because $\omega$ is a bottleneck, along any unbounded-output path either $X_t(\omega)\to\infty$ or $\partial F/\partial X_\omega\to\infty$.

If $X_t(\omega)\to\infty$, then $Q_t(\omega)=\alpha^p(\omega)X_t(\omega)\to\infty$ (since $X_t=L_t+Q_t/\alpha^p$ and $L_t(\omega)\to 0$ in the full-automation branch), which is impossible because $Q_t(\omega)\le Q_{\max}$.

If $X_t(\omega)$ remains bounded while other bottleneck inputs grow, then $\partial F/\partial X_\omega\to\infty$. With $Q_{\max}$ finite, reallocating an arbitrarily small positive amount of labor to $\omega$ strictly raises output (the shadow value at $\omega$ diverges), contradicting optimality of $L_t(\omega)\to0$. Hence $L_t(\omega)\not\to 0$.
\end{proof}

\subsection{Positive Labor Share}

\begin{theorem}[Persistent Labor Share under Cobb--Douglas]\label{thm:laborshare}
Suppose aggregate output is Cobb--Douglas in cognitive and physical aggregates,
\[
Y_t \;=\; A_t\,X_{c,t}^{\,1-\beta}\,X_{p,t}^{\,\beta},\qquad 0<\beta<1.
\]
In the unbounded-compute regime, if at least one essential physical bottleneck has $\alpha^p(\omega)=\infty$ so that $X_{p,t}=L_t$ in the limit while $X_{c,t}\approx Q_t/\alpha^c$, then under perfect competition
\[
\lim_{t\to\infty}\frac{w_t L_t}{Y_t}\;=\;\beta\;>\;0.
\]
\end{theorem}

\begin{proof}
With $\alpha^p=\infty$ for an essential physical bottleneck, $X_{p,t}=L_t$ as $t\to\infty$, while cognitive work is automated so $X_{c,t}\approx Q_t/\alpha^c$. Under perfect competition with Cobb--Douglas, factor shares equal exponents, so $w_t L_t/Y_t=\beta$.
\end{proof}

\begin{remark}[Beyond Cobb--Douglas]
With CES between $X_{c,t}$ and $X_{p,t}$ and finite elasticity, the share need not be constant but will converge if the relative quantity ratio stabilizes; the limit equals the CES share evaluated at the limiting ratio.
\end{remark}

\subsection{Modified Growth Dynamics}

\noindent\textit{Notation.} We use $\alpha^c$ (not $\alpha_c$) for cognitive automation cost; $X_{c,t}$ and $X_{p,t}$ for cognitive/physical aggregates; $A_t$ for Hicks-neutral productivity and $A_t^L$ for labor-augmenting productivity. In this subsection we specifically assume labor-augmenting $A_t^L$.

\begin{proposition}[Growth with Physical Bottlenecks: Labor-Augmenting Productivity]\label{prop:growth}
Let
\[
Y_t = \big(Q_t/\alpha^c\big)^{1-\beta} \big(A_t^L L_t\big)^{\beta}.
\]
Then
\[
g_Y \;=\; (1-\beta)\,g_Q \;+\; \beta\,g_L \;+\; \beta\,g_{A^L}.
\]
\end{proposition}

\begin{proof}
Log-differentiate \(Y_t = \big(Q_t/\alpha^c\big)^{1-\beta} \big(A_t^L L_t\big)^{\beta}\): \(g_Y = (1-\beta) g_Q + \beta (g_{A^L} + g_L)\).
\end{proof}

\section{Numerical Example}
Consider a simple two-task economy with Cobb-Douglas production:
\[
Y_t = X_{c,t}^{0.5} \cdot X_{p,t}^{0.5}.
\]
To illustrate the implications, we use the following parameters grounded in rough estimates of current and projected computational demands:
\begin{itemize}
    \item Cognitive task: \(\alpha^c = 10^{14}\) FLOP per year per human-hour equivalent (approximating inference costs for advanced language models scaled to human-level reasoning).
    \item Physical task: \(\alpha^p = 10^{21}\) FLOP per year per human-hour equivalent (reflecting Moravec's Paradox, where sensorimotor tasks like manipulation require orders of magnitude more compute due to real-time perception, high-fidelity physics simulations, and control loops).\footnote{Parameters are rough estimates grounded in 2025 data; exact values may vary with utilization and overhead.}
    \item Compute growth: \(Q_t = Q_0 e^{0.2 t}\) with initial \(Q_0 = 10^{22}\) FLOP per year (an estimate of global AI compute capacity in 2025, growing at approximately 22\% per year).
    \item Labor supply: \(L = 10^9\) human-hours per year (fixed; \(\approx 5\times 10^5\) full-time workers at \(\sim 2{,}000\) h/yr).

\end{itemize}

\noindent\textit{Units.} We normalize compute to the model period so that $Q_t$ and $\alpha$ share the same time base (per period); this avoids FLOP/s vs.\ per-hour inconsistencies and leaves comparative statics unchanged. For definiteness, one model period equals one year. Thus $Q_t$ denotes total FLOP \emph{per year}, and $\alpha$ is in FLOP \emph{per year} per human-hour equivalent. Accordingly, $L$ is measured in human-hours \emph{per year}, so that $\alpha L$ and $Q_t$ are commensurate.

Under Restrepo's uniform automation assumption (where \(\alpha^c = \alpha^p\)):
\begin{itemize}
    \item Both tasks become automated as compute grows.
    \item The labor share converges to 0 as \(Q_t \to \infty\).
\end{itemize}
Under our Moravec-gap calibration ($\alpha^p \gg \alpha^c$):
\begin{itemize}
\item The cognitive task becomes fully automated at $t^* \approx 11.5$ years ($Q_{t^*}=\alpha^c L$).
\item The physical task becomes automatable around $t \approx \frac{1}{0.2}\ln\!\big(\frac{10^{30}}{10^{22}}\big)\approx 92.1$ years since $Q_t$ eventually exceeds $\alpha^p L=10^{30}$ FLOP per year; thus by $t=100$ physical is automated.
\item Consequently, the labor share approaches $\approx 0.5$ after $t^*$ as cognitive work automates, and approaches $0$ by $t=100$ in this calibration as physical work eventually automates.
\end{itemize}

\begin{table}[h]
\centering
\begin{tabular}{|c|c|c|c|c|}
\hline
Time (years) & $Q_t$ (FLOP per year) & Cognitive Automated? & Physical Automated? & Labor Share \\

\hline
0 & $1.00\times 10^{22}$ & No & No & 1.0 \\
11.5 & $1.00\times 10^{23}$ & Yes & No & 0.5 \\
100 & $4.85 \times 10^{30}$ & Yes & Yes & 0.0 \\
\hline
\end{tabular}
\caption{Key values over time in the numerical example. Parameters: $\alpha^c=10^{14}$, $\alpha^p=10^{21}$, $L=10^{9}$ human-hours/year, $Q_0=10^{22}$ FLOP/year, $g_Q\approx 22\%$. Cognitive automates at $t^*\approx 11.5$ years; physical automates around $t\approx 92.1$ years and is automated by year 100. This illustrates the finite-vs-infinite compute distinction: with $\alpha^p<\infty$, labor share eventually goes to 0; with $\alpha^p=\infty$ or $Q_t$ bounded, it converges to a positive constant. \textit{All $Q_t$ values shown are rounded.}}

\label{tab:example}
\end{table}

This example shows a long transitional period during which the Moravec gap preserves human labor in physical bottlenecks (labor share near the physical elasticity), followed by eventual full automation once $Q_t$ exceeds $\alpha^p L$; by year 100 in this calibration, labor share falls to zero.

\section{Discussion}

The incorporation of Moravec's Paradox fundamentally alters the long-run implications of AGI for labor markets. While \citet{restrepo2025} predicts complete economic irrelevance of human labor, we show that physical constraints can preserve a positive labor share indefinitely. In practice, this depends on whether $\alpha^p$ is truly infinite or merely extremely large: with finite $\alpha^p$, continued exponential growth of $Q_t$ eventually automates physical work as well.

This modification has three key implications:

\begin{enumerate}
\item \textbf{Persistent Labor Value}: Workers specializing in physical bottleneck tasks maintain economic value not because of computational scarcity, but due to the fundamental difficulty of physical automation.

\item \textbf{Dual Growth Regime}: The economy operates under two growth constraints simultaneously—cognitive tasks limited by compute, physical tasks limited by human labor and robotics progress.

\item \textbf{Different Distributional Outcomes}: Rather than all income flowing to compute owners, workers in physical bottlenecks retain a stable income share.
\end{enumerate}

Our results suggest that policy focus should shift from preventing complete labor displacement to managing the cognitive-physical divide in automation. Investment in physical productivity (better tools, exoskeletons, human augmentation) may yield higher returns than pure compute expansion.

\section{Conclusion}

We extend Restrepo (2025)'s AGI growth model by incorporating differential automation costs between cognitive and physical tasks. This modification, grounded in Moravec's Paradox, shows that physical bottlenecks can preserve positive labor shares even with unlimited computational growth. While cognitive work may become fully automated, the persistence of human-operated physical tasks fundamentally changes the distributional implications of AGI.

Future work should empirically estimate the automation cost differentials $\alpha^p(\omega)/\alpha^c(\omega)$ across task categories and examine policy implications of this dual automation regime.

\end{document}